\documentclass[11pt,a4paper]{article}
\pdfoutput=1
\usepackage[left=1in,right=1in,top=1in,bottom=1in]{geometry}
\usepackage{lmodern} 
\usepackage[T1]{fontenc}
\usepackage[utf8]{inputenc}
\usepackage[tbtags]{amsmath}
\usepackage{amssymb}
\usepackage{amsthm}
\usepackage{hyperref}
\usepackage[numbers]{natbib}

\newtheorem{theorem}{Theorem}
\newtheorem{lemma}{Lemma}
\newtheorem{observation}[theorem]{Observation}
\newtheorem{corollary}[theorem]{Corollary}
\theoremstyle{definition}
\newtheorem{definition}{Definition}
\newtheorem{innercustomthm}{Lemma} 
\newenvironment{customlemma}[1] 
  {\renewcommand\theinnercustomthm{#1}\innercustomthm} 
  {\endinnercustomthm}

\newcommand{\E}{\mathbb{E}}

\def \R{\mathbb R}

\newcommand{\sset}[1]{\left\{ #1\right\}}
\newcommand{\ssets}[1]{\{ #1\}}
\newcommand{\fwh}[1]{\; \left| \; #1 \right.}
\newcommand{\fwhs}[1]{\; | \; #1 }

\newcommand{\cards}[1]{| #1 |}
\newcommand{\probability}[1]{\ensuremath{\mathrm{Pr}\left[#1\right]}}
\newcommand{\algoname}[1]{\ensuremath{\text{\rm\sc #1}}}

\DeclareMathOperator*{\expect}{\mathbb E}
\DeclareMathOperator*{\variance}{\mathrm{Var}}
\newcommand{\vecc}[1]{\ensuremath{\mathbf{#1}}}
\newcommand{\opt}{\ensuremath{\mathrm{OPT}}}
\title{The VCG Mechanism for Bayesian Scheduling\footnote{Supported by ERC Advanced Grant 321171 (ALGAME) and EPSRC grant EP/M008118/1.\newline\indent A preliminary version of this paper appeared in WINE'15~\citep{gkyr2015-wine}.}
}
\author{Yiannis Giannakopoulos\thanks{Department of Computer Science, University of Liverpool. Email: \href{mailto:ygiannak@cs.ox.ac.uk}{\nolinkurl{ygiannak@cs.ox.ac.uk}}
\newline\indent
A significant part of this work was carried out while the first author was a PhD student at the University of Oxford.}
\and
Maria Kyropoulou\thanks{Department of Computer Science, University of Oxford. Email: \href{mailto:kyropoul@cs.ox.ac.uk}{\nolinkurl{kyropoul@cs.ox.ac.uk} }}}
\date{March 28, 2017}
\begin{document}
\maketitle
\begin{abstract}
We study the problem of scheduling $m$ tasks to $n$ selfish, unrelated machines in order to minimize the makespan, where the execution times are independent random variables, identical across machines. We show that the VCG mechanism, which myopically allocates each task to its best machine, achieves an approximation ratio of $O\left(\frac{\ln n}{\ln \ln n}\right)$.  This improves significantly on the previously best known bound of $O\left(\frac{m}{n}\right)$ for prior-independent mechanisms, given by Chawla et al.\ [STOC'13] under the additional assumption of Monotone Hazard Rate (MHR) distributions. Although we demonstrate that this is in general tight, if we do maintain the MHR assumption, then we get improved, (small) constant bounds for $m\geq n\ln n$ i.i.d.\ tasks, while we also identify a sufficient condition on the distribution that yields a constant approximation ratio regardless of the number of tasks.
\end{abstract}

\section{Introduction}
We consider the problem of scheduling tasks to machines, where the processing
times of the tasks are \emph{stochastic} and the machines are \emph{strategic}.
The goal is to minimize the expected completion time (a.k.a.\ \emph{makespan})
of any machine, where the expectation is taken over the randomness of the
processing times and the possible randomness of the mechanism. We are interested
in the performance, i.e.\ the expected makespan, of \emph{truthful
mechanisms} compared to the optimal \emph{algorithm} that does not take the
incentives of the machines into consideration. This problem, which we call the
\emph{Bayesian scheduling} problem, was previously considered by \citet{Chawla:2013rm}.
Scheduling problems constitute a very rich and intriguing area of research \cite
{HochbaumHall97}. In one of the most fundamental cases, the goal is to schedule
$m$ tasks to $n$ parallel machines while minimizing the makespan, when the
processing times of the tasks are selected by an adversary in an arbitrary way
and can depend on the machine to which they are allocated. However, the
assumption that the machines will blindly follow the instructions of a central
authority (scheduler) was eventually challenged, especially due to the rapid
growth of the Internet and its use as a primary computing platform. Motivated by
this, in their seminal paper \citet{Nisan:2001aa} introduced a mechanism-design
approach to the scheduling problem: the processing times of the tasks are now
private information of the machines, and each machine declares to the mechanism how much time it requires to execute each task. The mechanism then outputs the allocation of tasks to machines, as well as monetary compensations to the machines for their work, based solely on these declarations. In fact, the mechanism has to decide the output in advance, for any possible matrix of processing times the machines can report. Each machine is assumed to be rational and strategic, so, given the mechanism and the true processing times, its declarations are chosen in order to minimize the processing time/cost it has to spend for the execution of the allocated tasks minus the payment it will receive. In this scenario, the goal is to design a \emph{truthful} mechanism that minimizes the makespan; truthful mechanisms define the allocation and payment functions so that the machines don't have an incentive to misreport their true processing-time capabilities. We will refer to this model as the \emph{prior free} scheduling problem, as opposed to the stochastic model we discuss next.

In the \emph{Bayesian} scheduling problem \cite{Chawla:2013rm}, the time a
specific machine requires in order to process a task is drawn from a distribution.
We consider one of the fundamental questions posed by the algorithmic mechanism
design literature, which is about quantifying the potential performance loss of
a mechanism due to the requirement for truthfulness. In the Bayesian scheduling
setting, this question translates to: \emph{What is the maximum ratio (for any
distribution of processing times) of the expected makespan of the best
truthful mechanism over the expected optimal makespan (ignoring the machines'
incentives)?}

In this paper we tackle this question by considering a well known and natural
truthful mechanism, the \emph{Vickrey-Clarke-Groves mechanism (VCG)} \cite
{Vickrey61,Clarke71,Groves73}. VCG can be defined for very general mechanism
design settings. In the special case of scheduling
unrelated machines, it has a very simple interpretation:
greedily and
myopically allocate each task to a machine that minimizes its processing time.
It is a well known fact that VCG is a truthful mechanism in a very strong sense;
truth-telling is a dominant strategy for the machines. Because of the notorious lack of characterization results for truthfulness for restricted domains such as scheduling, VCG (or more generally, affine maximizers) is the standard and obvious choice to consider for the Bayesian scheduling problem. We stress here that for the scheduling domain (and for any additive domain) the VCG allocation and payments can be computed in polynomial time. Also, it is important to note that VCG is a \emph{prior-independent} mechanism, i.e.\ it does not require any knowledge of the prior distribution from which the processing times are drawn.

Prior-independence is a very strong property, and is an important feature for mechanisms used in stochastic settings. Being robust with respect to prior distributions facilitates applicability in real systems, while at the same time bypasses the pessimistic inapproximability results of worst case analysis. The idea is that we would like the mechanisms we use, without relying on any knowledge of the distribution of the processing times of the tasks, to still perform well compared to the optimal mechanism that is tailored for the particular distribution.

\citet{Chawla:2013rm} were the first to examine the Bayesian scheduling problem while considering the importance for prior-independence. They study the following two mechanisms:
\begin{description}
\item[Bounded overload with parameter $c$] Allocate tasks to machines such that
the sum of the processing times of all tasks is minimized, subject to placing at most $c\frac{m}{n}$ tasks at any machine.
\item[Sieve and bounded overload with parameters $c,\beta,$ and $\delta$] Fix a
partition of the machines into two sets of sizes $(1-\delta)n$ and $\delta n$. Ignoring all processing times which exceed\footnote{Assume you run VCG on the first set of machines plus a dummy machine with processing time $\beta$ on all tasks. The case where a task has processing time equal to $\beta$ can be ignored without loss of generality for the case of continuous distributions.} $\beta$ (i.e.\ setting them equal to infinity), run VCG on the first set of machines. For the tasks that remain unallocated  run the bounded overload mechanism with parameter $c$ on the second set of machines.
\end{description}

The above mechanisms are inspired by maximal-in-range (affine maximizer)
mechanisms \cite{NR07} and threshold mechanisms, as these are essentially the
only non-trivial truthful mechanisms we know for the scheduling domain.
One would expect that the simplest of these mechanisms, which is the VCG
mechanism, would be the first to be considered. Indeed, VCG is the most natural,
truthful, simple, polynomial time computable, and prior-independent mechanism.
Still, the authors in \cite{Chawla:2013rm} design the above mechanisms in an
attempt to prevent certain bad behaviour that VCG exhibits on a specific input
instance and don't examine VCG beyond that point. As we demonstrate in this
paper, however, this specific instance actually constitutes the worst case
scenario for VCG
and we can identify cases where VCG performs considerably better, either by placing a restriction on the number of tasks or by making some additional distributional assumptions.

\paragraph{Our results.}
We prove an asymptotically tight bound of $\Theta\left(\frac{\ln n}{\ln \ln n}\right)$ for the approximation ratio of VCG for the Bayesian scheduling problem under the sole assumption that the machines are a priori identical. This bound is achieved by showing that the worst case input for VCG is actually one where the tasks are all of unit weight (point mass distributions). This resembles a balls-in-bins type scenario from which the bound is implied.

Whenever the processing times of the tasks are i.i.d. and drawn from an MHR continuous distribution, VCG is shown to be $2\left(1+\frac{n\ln n}{m}\right)$-approximate for the Bayesian scheduling problem. This immediately implies a constant bound at most equal to $4$ when $m\geq n\ln n$. We also get an improved bound of $1+\sqrt{2}$ when $m\geq n^2$ using a different approach. For the complementary case of $m\leq n\ln n$, we identify a property of the distribution of processing times such that VCG again achieves a constant approximation.
We observe that important representatives of the class of MHR distributions, that is the uniform distribution on $[0,1]$ as well as exponential distributions, do satisfy this property, so for these distributions VCG is $4$-approximate regardless of the number of tasks. We note however that this is not the case for all MHR distributions.

The continuity assumption plays a fundamental role in the above results. In particular, we give a lower bound of $\Omega\left(\frac{\ln n}{\ln \ln n}\right)$ for the case of i.i.d. processing times that uses a discrete MHR distribution. Finally, we also consider the bounded overload and the sieve and bounded overload mechanisms that were studied by \citet{Chawla:2013rm}, and present some instances that lower-bound their performance.

\paragraph{Related Work.}
One of the fundamental papers on the approximability of scheduling with
unrelated machines is by \citet{LST90} who provide a polynomial time algorithm
that approximates the optimal makespan within a factor of $2$. They also prove
that it is NP-hard to approximate the optimal makespan within a factor of $3/2$
in this setting. In the mechanism design setting, \citet{Nisan:2001aa} prove
that the well known VCG mechanism achieves an $n$-approximation of the optimal
makespan, while no truthful mechanism can achieve approximation ratio better
than $2$. Note that the upper bound immediately carries over to the Bayesian and
the prior-independent scheduling case. The lower bound has been improved by
\citet{CKV09} and \citet{KV13} to $2.61$, while \citet{ADL12} prove the
tightness of the upper bound for deterministic anonymous mechanisms. In contrast
to the negative result on the prior free setting presented in \cite{ADL12}, truthful mechanisms can achieve sublinear approximation when the processing times are stochastic. In fact, we prove here that VCG can achieve a sublogarithmic approximation, and even a constant one for some cases, while similar bounds for other mechanisms have also been presented by \citet{Chawla:2013rm}.

For the special case of related machines, where the private information of each
machine is a single value, \citet{AT01} were the first to give a
$3$-approximation truthful in expectation mechanism, while by now truthful PTAS
are also known~\cite{CK13,DDDR11,Epstein2016}. Putting
computational
considerations aside, the best truthful mechanism in this single-dimensional setting is also optimal. \citet{LS09} managed to prove constant approximation ratio for a special, yet multi-dimensional scheduling problem; they consider the case where the processing times of each task can take one of two fixed values. \citet{Y09} then generalized this result to two-range-values, while together with \citet{LY08} and \citet{L09}, they gave constant (better than $1.6$) bounds for the case of two machines.

\citet{DW15} consider computationally tractable approximations with respect to
the best
(Bayesian) truthful mechanism when the processing times of the tasks follow
distributions (with finite support) that are known to the mechanism designer. In fact the authors provide a reduction of this problem to an algorithmic problem. \citet{CIL12} showed that there can be no approximation-preserving reductions from mechanism design to algorithm design for the makespan objective, however the authors in \cite{DW15} bypass this inapproximability by considering the design of bi-criterion approximation algorithms.

Prior-independent mechanisms have been mostly considered in the context of
optimal auction design, where the goal is to design an auction mechanism that
maximizes the seller's revenue. Inspired by the work of \citet{DRY15}, \citet{DHKN11} and \citet{RTY12} independently provide approximation mechanisms for
multi-dimensional settings, with recent follow-up work by \citet{Goldner2016} and \citet{Azar:2014}.
Moreover, \citet{DRS12} identify conditions under which VCG obtains a constant fraction of the optimal revenue, while \citet{HR09} prove Bulow-Klemperer type results for VCG. Prior robust optimization is also discussed by \citet{S13}.

\citet{Chawla:2013rm} are the first to consider \emph{prior-independent} mechanisms for the (Bayesian) scheduling problem. They introduce two variants of the VCG mechanism and bound their approximation ratios.
In particular, the bounded overload mechanism is prior-independent and achieves a $O(\frac{m}{n})$ approximation of the expected optimal makespan when the processing times of the tasks are drawn from machine-identical MHR distributions. For the case where the processing times of the tasks are i.i.d.\ from an MHR distribution, the authors prove that sieve and bounded overload mechanisms can achieve an $O(\sqrt{\ln n})$ approximation of the expected optimal makespan,
as well as an approximation ratio of $O((\ln \ln n)^2)$ under the additional assumption that there are at least $n\ln n$ tasks. We note that to achieve these improved approximation ratios, a sieve and bounded overload mechanism needs to have access to a small piece of information regarding the distribution of the processing times, in particular the expectation of the minimum of a certain number of draws (in contrast to VCG which requires no distributional information whatsoever).

The VCG mechanism is strongly represented in the above works. Its simplicity
and amenability to practice strongly motivate a detailed analysis of its
performance for the Bayesian scheduling problem. From our results, it turns
out that in general VCG performs better than the previously analyzed
prior-independent mechanisms, applies to wider settings with less restrictions
on the distributions and, of course, it is simpler. To summarize and clarify
this comparison with the previous prior-free mechanisms of
\citet{Chawla:2013rm}, we note that the only case where VCG demonstrates a
worse approximation ratio is when the number of tasks is asymptotically very
close to that of machines, in particular $m=o\left(\frac{\ln n}{\ln\ln
n}n\right)$ and, in addition, we are in a restricted setting where the
execution times have to be drawn from necessarily \emph{non-identical, MHR}
distributions. For
example, for
$m=n$ tasks with processing times drawn from machine-identical MHR
distributions which however differ across tasks, the bounded overload mechanism
of \citet{Chawla:2013rm} would
be constant $O(1)$-approximate, while VCG would have an approximation ratio of
$\Theta\left(\frac{\ln n}{\ln \ln n}\right)$. However, a point worth
mentioning here is that the constant hidden within the $O(1)$-notation above
is $800$ while the one in the upper-bound $O\left(\frac{\ln n}{\ln \ln
n}\right)$ of VCG comes directly from a balls-in-bins analysis and therefore is $1+o(1)$.

\section{Preliminaries and Notation}\label{sec:preliminaries}
Assume that we have $n$ unrelated parallel machines and $m\geq n$ tasks that
need to be scheduled to these machines. Let $t_{ij}$ denote the processing time of task $j$ on machine $i$. In the Bayesian scheduling problem, each $t_{ij}$ is independently drawn from some probability distribution $\mathcal D_{ij}$. In this paper we mainly consider the machine-identical setting, that is the processing times of a specific task $j$  are drawn from the same distribution $D_j$ for all the machines. This is a standard assumption for the problem (see also \cite{Chawla:2013rm}). We also consider the case where both machines and tasks are considered a priori identical, and the processing times $t_{ij}$ are all i.i.d. drawn from the same distribution $\mathcal{D}$. The goal is to design a truthful mechanism that minimizes the expected makespan of the schedule.

We consider the VCG mechanism, the most natural and standard choice for a truthful mechanism. Thus, we henceforth assume that the machines always declare their true processing times. VCG minimizes the total workload by allocating each task to the machine that minimizes its processing time. So, if $\mathbf{\alpha}$ denotes the allocation function of VCG (we omit the dependence on $\vecc t$ for clarity of presentation) then, for any task $j$, $\alpha_{ij}=1$ for some machine $i$ such that $t_{ij}=\min_{i'}\{t_{i'j}\}$, otherwise $\alpha_{ij}=0$. Without loss of generality we assume that in case of a tie, the machine is chosen uniformly at random\footnote{We note here that for continuous distributions, such events of ties occur with zero probability.}. The expected makespan of VCG is then computed as
$$
\E\left[\algoname{VCG}(\vecc t)\right] =\E\left[\max_i\sum_{j=1}^m\alpha_{ij}t_{ij} \right].
$$
In what follows, we use variable $Y_{i,j}$ to denote the processing time of task $j$ on machine $i$ under VCG, that is $Y_{i,j}=\alpha_{ij}t_{ij}$. We also denote by $Y_{i}=\sum_{j=1}^mY_{i,j}$ the workload of machine $i$.

Note that in the machine-identical setting $\alpha_{ij}=1$ with probability
$\frac{1}{n}$ for any task $j$. So, VCG exhibits a balls-in-bins type behaviour
in this setting, as the machine that minimizes the processing time of each task,
and hence, the machine that will be allocated the task, is chosen uniformly at
random for each task. We thus know from traditional balls-in-bins analysis, that
the expected maximum number of tasks that will be allocated to any machine will
be $\Theta\left(\frac{\ln n}{\ln \ln n}\right)$, whenever $m=\Theta(n)$. For
more precise balls-in-bins type bounds see \citet{Raab:1998ab}. We will use the
following theorem to prove in Section~\ref{sec:upper-bounds} that the above
instance that yields the $\Theta\left(\frac{\ln n}{\ln \ln n}\right)$ bound is
actually the worst case scenario for VCG:

\begin{theorem}[\citet{Berenbrink:2008ab}]
\label{th:max_load_majorization}
Assume two vectors $\vecc w\in \R^{m}$, $\vecc w'\in \R^{m'}$ with $m\leq m'$ and their values in non-increasing order (that is $w_1\geq w_2\geq \ldots \geq w_m$ and $w_1'\geq w_2'\geq \ldots \geq w'_{m'}$). If the following two conditions hold:
\begin{itemize}
\item[(i)] $\sum_{j=1}^m w_j = \sum_{j=1}^{m'} w_j'$
\item[(ii)] $\sum_{j=1}^k w_j\geq \sum_{j=1}^k w_j'$ \quad for all $k\in[m]$,
\end{itemize}
then the expected maximum load when allocating $m$ balls with weights according to $\vecc w$ is at least equal to the expected maximum load when allocating $m'$ balls with weights according to $\vecc w'$, uniformly at random to the same number of bins.
\end{theorem}
Following \cite{Berenbrink:2008ab} we say that vector $\vecc w$ \emph{majorizes}
$\vecc w'$ whenever $\vecc w$ and $\vecc w'$ satisfy Conditions (i) and (ii) of Theorem~\ref{th:max_load_majorization}.

\paragraph{Probability preliminaries.} We now give some additional notation regarding properties of distributions that will be used in the analysis.

Let $T$ be a random variable following a probability distribution $\mathcal D$. Assuming we perform $n$ independent draws from $\mathcal D$, we use  $T[r:n]$ to denote the $r$-th order statistic (the $r$-th smallest) of the resulting values, following the notation from~\cite{Chawla:2013rm}. In particular, $T[1:n]$ will denote the minimum of $n$ draws from $\mathcal D$, while $T[1:n][m:m]$ denotes the maximum value of $m$ independent experiments where each one is the minimum of $n$ draws from $\mathcal{D}$. Note that for $t_{ij}\sim \mathcal{D}_j$, the expected processing time of machine $i$ for task $j$ under VCG is
\begin{equation}
\label{eq:VCG-expect-helper1}
\expect[Y_{i,j}]=\probability{\alpha_{ij}=1}\expect\left[t_{ij}\fwh{\alpha_{ij}=1}\right]=\frac{1}{n}\expect[T[1:n]].
\end{equation}

In this work we also consider the class of probability distributions that have a \emph{monotone hazard rate (MHR)}. A continuous distribution with pdf $f$ and cdf $F$ is MHR if its hazard rate $h(x)=\frac{f(x)}{1-F(x)}$ is a (weakly) increasing function. The definition of discrete MHR distributions is similar, only the hazard rate of a discrete distribution is defined as $h(x)=\frac{\probability{X= x}}{\probability{X\geq x }}$ (see e.g.\ \citet{Barlow:1963fk}).  The following two technical lemmas demonstrate properties of MHR distributions. The proofs can be found in Appendix~\ref{append:preliminaries}.

\begin{lemma}
\label{lemma:first_order_statistic_MHR}
If $T$ is a continuous MHR random variable, then for every positive integer $n$, its first order statistic $T[1:n]$ is also MHR.
\end{lemma}

\begin{lemma}
\label{lemma:MHR_squares}
For any \emph{continuous} MHR random variable $X$ and any positive integer $r$, $\expect[X^r]\leq r!\expect[X]^r$.
\end{lemma}

We now introduce the notion of $k$-stretched distributions. The property that identifies these distributions plays an important role in the approximation ratio of VCG as we will see later in the analysis (Theorem~\ref{lem:technical_distr_ass}).
\begin{definition}\label{def:stretched}
Given a function $k$ over integers, we call a distribution \emph{$k$-stretched} if its order statistics satisfy
$$\expect[T[1:n][n:n]]\geq k(n)\cdot  \expect[T[1:n]],$$
for all positive integers $n$.
\end{definition}

We will use the following result by Aven to bound the expected makespan of VCG.
\begin{theorem}[\citet{Aven:1985kx}]
\label{th:bounds_max_mean_variance}
If $X_1,X_2,\dots,X_n$ are (not necessarily independent) random variables with mean $\mu$ and variance $\sigma^2$, then
$$
\expect[\max_i X_i]\leq \mu+\sqrt{n-1}\sigma.
$$
\end{theorem}

Finally, we use the notation introduced in the probability preliminaries to present some known bounds on the expected optimal makespan. So, if given a matrix of processing times $\vecc t$ we denote its optimal makespan by $\opt(\vecc t)$, we wish to bound $\E_{\vecc t}\left[\opt(\vecc t)\right]$ (we omit dependence on $\vecc t$ for clarity of presentation). Part of the notorious difficulty of the scheduling problem stems exactly from the lack of general, closed-form formulas for the optimal makespan. However, the following two easy lower bounds are widely used (see e.g. \cite{Chawla:2013rm}):

\begin{observation}
\label{lemma:trivial_bounds_opt}
If the processing times are drawn from machine-identical distributions, then the
expected optimal makespan is bounded by
$$
\E[\opt]\geq \max\sset{\expect\left[\max_j T_j[1:n]\right],\frac{1}{n}\sum_{j=1}^m\expect\left[T_j[1:n]\right]},
$$
where $T_j$ follows the distribution corresponding to task $j$.
\end{observation}

\section{Upper Bounds}
\label{sec:upper-bounds}
In this section we provide results on the performance of the VCG mechanism for the Bayesian scheduling problem for different assumptions on the number of tasks (compared to the machines), and different distributional assumptions on their processing times. Our first result shows that VCG is $O\left(\frac{\ln n}{\ln \ln n}\right)$--approximate in the general case, without assuming identical tasks or even MHR distributions. We then consider some additional assumptions under which VCG achieves a constant approximation of the expected optimal makespan. In what follows, an allocation where all machines have the same workload will be called \emph{fully balanced}.

\begin{theorem}\label{thm:ub_identical_machines}
VCG is $O\left(\frac{\ln n}{\ln \ln n}\right)$-approximate for the Bayesian scheduling problem with $n$ identical machines.
\end{theorem}
As we will see later in Theorem~\ref{thm:lb-vcg-identical-machines}, this result is in general tight. In order to prove Theorem~\ref{thm:ub_identical_machines} we will make use of the following lemma:
\begin{lemma}
\label{lem:VCG_reduction_worst-case}
If \algoname{VCG} is $\rho$-approximate for the prior free scheduling problem with identical machines on inputs for which the optimal allocation is fully balanced, then \algoname{VCG} is $\rho$-approximate for the Bayesian scheduling problem where the machines are a priori identical.
\end{lemma}
\begin{proof}
We will show that for any instance of Bayesian scheduling with a priori identical machines, there exists a prior free scheduling instance with identical machines for which the approximation ratio of VCG is at least the same. In fact, there exists such a prior free instance, for which the optimal allocation is fully balanced.

Consider a Bayesian scheduling instance where $t_{ij}\sim \mathcal{D}_j$ for tasks $j\in[m]$ and machines $i\in [n]$. Let $\rho\geq 1$ be such that $\expect_{\vecc t}[\algoname{VCG}(\vecc t)]=\rho\cdot \E_{\vecc t}[\opt(\vecc t)]$. Then, conditioning on the minimum processing times of the tasks, there exists an $m$-dimensional vector $(t^*_1,\dots,t^*_m)$ such that
$$\expect_{\vecc t}\nolimits\left[\algoname{VCG}(\vecc t)\fwh{\min_i t_{i1}=t^*_1\land\dots\land \min_i t_{im}=t^*_m}\right]\geq\rho\cdot \expect_{\vecc t}\nolimits\left[\opt(\vecc t)\fwh{\min_i t_{i1}=t^*_1\land\dots\land \min_i t_{im}=t^*_m}\right].$$
Notice that, once such a minimum processing time $t_j^*$ has been fixed for
all tasks, the only randomization remaining within the expected makespan of
$\algoname{VCG}$ is the one with respect to the identities of the machines having processing time $t^*_j$ and the possible internal tie breaking; thus, if
we let $\vecc t^*$ denote the time matrix where task $j$ has processing time
$t_{ij}=t_j^*$ for all machines $i$, it holds that  $$\expect_{\vecc
t}\nolimits\left[\algoname{VCG}(\vecc t)\fwh{\min_i t_{i1}=t^*_1\land\dots\land \min_i t_{im}=t^*_m}\right]=\algoname{VCG}(\vecc t^*).$$
Also, once we have fixed the smallest element in every column of an input matrix $\vecc t$ (a column contains the processing times of a single task on all the machines), reducing all other values of a column $j$ to be equal to that minimum value $t_j^*$ can only improve the optimal makespan, so
$$\expect_{\vecc t}\nolimits\left[\opt(\vecc t)\fwh{\min_i t_{i1}=t^*_1\land\dots\land \min_i t_{im}=t^*_m}\right]\geq \opt(\vecc t^*).$$
Combining the above, we get that indeed $\algoname{VCG}(\vecc t^*)\geq\rho\cdot \opt(\vecc t^*)$.

It remains to be shown that, without loss, $\vecc t^*$ gives rise to an optimal (prior free) allocation that is fully balanced, that is all machines have exactly the same workload (equal to the optimal makespan).  Indeed, if that is not the case, then for any machine whose workload is strictly below the optimal makespan, we can slightly increase the processing time $t^*_j$ of one of its tasks $j$ without affecting the optimal makespan, while at the same time that increase can only make the performance of $\algoname{VCG}$ worse.
\end{proof}

We are now ready to prove Theorem~\ref{thm:ub_identical_machines}. Lemma~\ref{lem:VCG_reduction_worst-case} essentially reduces the analysis of \algoname{VCG} for the Bayesian scheduling problem for identical machines to that of a simple weighted balls-in-bins setting:

\begin{proof}[Proof of Theorem~\ref{thm:ub_identical_machines}]
From Lemma~\ref{lem:VCG_reduction_worst-case}, it is enough to analyze the performance of VCG on input matrices where the processing time of each task is the same across all machines and the optimal schedule is fully balanced. Without loss (by scaling) it can be further assumed that the optimal makespan is exactly $1$. Then, since VCG is breaking ties uniformly at random, the problem is reduced to analyzing the expected maximum (weighted) load when throwing $m$ balls with weights $(w_1,\dots,w_m)=\vecc w$ (uniformly at random) into $n$ bins, when $\sum_{j=1}^m w_j=n$.
Then, by Theorem~\ref{th:max_load_majorization}, that maximum load is upper bounded by the expected maximum load of throwing $n$ (unit weight) balls into $n$ bins, because the $n$-dimensional unit vector $\vecc 1_n$  majorizes $\vecc w$: $\vecc 1_n$'s components sum up to $n$ and also $w_j\leq 1$ for all $j\in [n]$ (due to the assumption that the optimal makespan is $1$).
By classic balls-in-bins results (see e.g.~\cite{Motwani:1995aa,Raab:1998ab}), the expected maximum load of any machine is upper bounded by $\varTheta\left(\frac{\ln n}{\ln \ln n}\right)$.
\end{proof}

We now focus on the special but important case where both tasks and machines are a priori identical:
\begin{theorem}\label{thm:bound4}
VCG is $2\left(1+\frac{n\ln n}{m}\right)$-approximate for the Bayesian scheduling problem with i.i.d.\ processing times drawn from a continuous MHR distribution.
\end{theorem}
\begin{proof}
Let $T$ be a random variable following the distribution from which the execution
times $t_{ij}$ are drawn. Following the notation introduced in the introduction, the workload of a machine $i$ is given by the random variable $Y_i=\sum_{j=1}^m Y_{i,j}$. Then, for the expected makespan $\expect[\max_i Y_i]$ and any real $s>0$ it holds that
\begin{align}\label{eq:exp_max}
e^{s\cdot\expect[\max_i Y_i]}&\leq \expect[e^{s \max_i Y_i}] = \expect[\max_i e^{s Y_i}] \leq \sum_{i=1}^n \expect[e^{s Y_i}]=\sum_{i=1}^n \prod_{j=1}^m \expect[ e^{s Y_{i,j}}]= n \expect[ e^{s Y_{1,1}}]^m,
\end{align}
where we have used Jensen's inequality based on the convexity of the exponential
function, and the fact that for a fixed machine $i$ the random variables $Y_{i,j}$, $j=1,\dots,m$, are independent (the processing times are i.i.d. and VCG allocates each task independently of the others). We now bound the term $\expect[ e^{s Y_{1,1}}]$:
\begin{align*}
\expect[ e^{s Y_{1,1}}]&=\expect\left[\sum_{r=0}^\infty \frac{(s Y_{1,1})^r}{r!}\right]=1+\sum_{r=1}^\infty s^r\frac{ \expect[Y_{1,1}^r]}{r!}=1+\frac{1}{n}\sum_{r=1}^\infty s^r\frac{ \expect[T[1:n]^r]}{r!}\leq 1+\frac{1}{n}\sum_{r=1}^\infty s^r\expect[T[1:n]]^r,
\end{align*}
where for the last inequality we have used the fact that the first order statistic of an MHR distribution is also MHR (Lemma~\ref{lemma:first_order_statistic_MHR}) and Lemma~\ref{lemma:MHR_squares}. Then, by choosing $s=s^*\equiv\frac{1}{2\cdot\expect[T[1:n]]}$  we get that
\begin{align*}
\expect[ e^{s^* Y_{1,1}}]\leq 1+\frac{1}{n}\sum_{r=1}^\infty \frac{1}{2^r} \leq  1+\frac{1}{n},
\end{align*}
and (\ref{eq:exp_max}) yields
\begin{align}\nonumber
\expect[\max_i Y_i]&\leq \ln\left(n \expect[ e^{s^* Y_{1,1}}]^m\right)\frac{1}{s^*}\\\nonumber
&\leq 2\ln\left(n \left(1+\frac{1}{n}\right)^m\right)\expect[T[1:n]]\\\nonumber
&\leq 2\ln\left(n e^{m/n}\right)\expect[T[1:n]]\\\label{eq:exp_makespan}
&=2\left( \ln n+\frac{m}{n}\right)\expect[T[1:n]].
\end{align}
But from Observation~\ref{lemma:trivial_bounds_opt} we know that $\E[\opt] \geq \frac{m}{n}\expect[T[1:n]]$ for the case of i.i.d.\ execution times, and the theorem follows.
\end{proof}

Notice that Theorem~\ref{thm:bound4} in particular implies that VCG achieves a \emph{small, constant} approximation ratio whenever the number of tasks is slightly more than that of machines:

\begin{corollary}
\label{cor:bound4}
VCG is $4$-approximate for the Bayesian scheduling problem with
$m\geq n \ln n$ i.i.d.\ tasks drawn from a continuous MHR distribution.
\end{corollary}

The following theorem will help us analyze the performance of VCG for the complementary case to that of Corollary~\ref{cor:bound4}, that is when the number of tasks is $ m\leq n \ln n$. Recall the notion of $k$-stretched distributions introduced in Definition~\ref{def:stretched}.
\begin{theorem}\label{lem:technical_distr_ass}
VCG is $4\frac{\ln n}{k(n)}$-approximate for the Bayesian scheduling problem with $m\leq n \ln n$ i.i.d.\ tasks drawn from a $k$-stretched MHR distribution.
\end{theorem}
\begin{proof}
From \eqref{eq:exp_makespan} we can deduce that the approximation ratio of VCG
is upper bounded by
\begin{align*}
\left( \ln n+\frac{m}{n}\right)\frac{2\expect[T[1:n]]}{\E[\opt]}
&\leq \left( \ln n+\frac{m}{n}\right)\frac{2\expect[T[1:n]]}{\expect[T[1:n][m:m]]}\\
&\leq \left( \ln n+\frac{m}{n}\right)\frac{2\expect[T[1:n]]}{\expect[T[1:n][n:n]]}\\
&\leq \left( \ln n+\frac{m}{n}\right)\frac{2}{k(n)}\\
&\leq \frac{4\ln n}{k(n)},
\end{align*}
where we have used Observation~\ref{lemma:trivial_bounds_opt} and the fact that $n\leq m\leq n\ln n$.
\end{proof}
In particular, we note that Theorem~\ref{lem:technical_distr_ass} yields a constant approximation ratio for VCG for the important special cases where the processing times are drawn independently from the uniform distribution on $[0,1]$ or any exponential distribution. Indeed, the uniform distribution on $[0,1]$ as well as any exponential distribution is $\ln$-stretched. See Appendix~\ref{append:expo_uniform_computations} for a full proof. We get the following, complementing the results in Corollary~\ref{cor:bound4}:
\begin{corollary}\label{cor:exp_uniform}
VCG is $4$-approximate for the Bayesian scheduling problem with i.i.d.\ processing times drawn from the uniform distribution on $[0,1]$ or an exponential distribution.
\end{corollary}
We point out that the above corollary can not be generalized to hold for all MHR distributions, as the lower bound in Theorem~\ref{thm:lb-vcg-identical-machines} implies. For example, it is not very difficult to check that by taking $\varepsilon\to 0$ and considering the uniform distribution over $[1,1+\varepsilon]$, no stretch factor $k(n)=\varOmega(\ln n)$ can be guaranteed.

For our final positive result, we present an improved constant bound on the approximation ratio of VCG when we have many tasks:
\begin{theorem}\label{th:VCG_upper_iid_const}
VCG is $1+\sqrt{2}$-approximate for the Bayesian scheduling problem with $m\geq n^2$ tasks with i.i.d. processing times drawn from a continuous MHR distribution.
\end{theorem}
\begin{proof}
We use Theorem~\ref{th:bounds_max_mean_variance} to bound the performance of VCG in this setting. In order to do so, we first bound the expectation and the variance of the makespan of a single machine. From~\eqref{eq:VCG-expect-helper1}, for the workload $Y_{i}$ of any machine $i$ we have:
\begin{equation*}
\expect[Y_i]= \sum_{j=1}^m\expect[Y_{i,j}]
			=\frac{1}{n}\sum_j\expect[T[1:n]]=\frac{m}{n}\expect[T[1:n]].
\end{equation*}
To compute the variance of the makespan of machine $i$, we note that the random variables $Y_{i,j}$ are independent with respect to $j$, for any fixed machine $i$ and thus we can get
\begin{align*}
\variance[Y_i] &=\sum_{j=1}^m\variance[Y_{i,j}] =\sum_{j=1}^m\left(\expect[Y_{i,j}^2]-\expect[Y_{i,j}]^2\right)\\
		&\leq \sum_{j=1}^m \expect[Y_{i,j}^2]=\sum_{j=1}^m \expect[\alpha_{ij}^2t_{ij}^2]=\frac{1}{n}\sum_{j=1}^m\expect[{T[1:n]}^2]\\
		&=\frac{m}{n}\expect[{T[1:n]}^2].
\end{align*}

We are now ready to use Theorem~\ref{th:bounds_max_mean_variance} and bound the expected makespan:
\begin{align*}
\expect[\max_i Y_i]&\leq \expect[Y_1]+\sqrt{n-1}\sqrt{\variance[Y_1]}\\
&\leq \frac{m}{n}\expect[T[1:n]]+\sqrt{m}\sqrt{\expect[{T[1:n]}^2]}\\
&\leq \frac{m}{n}\expect[T[1:n]]+ \sqrt{2}\sqrt{m}\expect[{T[1:n]}]\\
&\leq (1+\sqrt{2})\frac{m}{n}\expect[T[1:n]]\\
&\leq (1+\sqrt{2})\E[\opt],
\end{align*}
where the third inequality follows from Lemma~\ref{lemma:MHR_squares} (and Lemma~\ref{lemma:first_order_statistic_MHR}), for the fourth inequality we use the assumption that $m\geq n^2$ and to complete the proof, the last inequality uses a lower bound on $\E[\opt]$ from Observation~\ref{lemma:trivial_bounds_opt}.
\end{proof}

\section{Lower Bounds}
In this section we prove some lower bounds on the performance of VCG under
different distributional assumptions on the processing times. In an attempt for
a clear comparison of VCG with the mechanisms that were previously considered for
the Bayesian scheduling problem (in \cite{Chawla:2013rm}), we provide instances that lower bound
their performance as well.
\begin{theorem}\label{thm:lb-vcg-identical-machines}
For any number of tasks, there exists an instance of the Bayesian scheduling problem where VCG is not better than $\Omega\left(\frac{\ln n}{\ln \ln n}\right)$-approximate and the processing times are drawn from machine-identical continuous MHR distributions.
\end{theorem}
\begin{proof}
Consider an instance with $n$ identical machines and $m$ tasks where for any machine $i$, task $j$ has processing time $t_{ij}=1$ with probability $1$ for $j=1,\dots,n-1$ and processing time $t_{ij}=\frac{1}{m-n+1}$ with probability $1$ for $j=n,\ldots,m$.
From classical results from balls-in-bins analysis (see also the proof of Theorem~\ref{thm:ub_identical_machines}) we can deduce that the expected maximum number of unit-weight tasks allocated to any machine by VCG, is $\varOmega\left(\frac{\ln n}{\ln \ln n}\right)$. On the other hand, there exists an allocation that achieves a makespan equal to $1$, that is to allocate all of the $m-n+1$ ``small'' tasks to a single machine and allocate each of the remaining unit-cost tasks to a different machine.
The theorem follows by noticing that we can without loss replace these point-mass distributions on $1$ and $\frac{1}{m-n+1}$  with uniform distributions over small intervals around that points.
\end{proof}

Notice that when the number of tasks equals that of the machines, i.e.\ $m=n$,
then the bad instance for the lower bound of Theorem~\ref
{thm:lb-vcg-identical-machines} is in fact an i.i.d.\ instance where tasks
are identical as well and all $t_
{ij}$'s are drawn from the same distribution, and not just an instance with
only machines being identical. However, if we restrict our focus only on
discrete
distributions, then we can strengthen that lower bound to hold for i.i.d.\ distributions for essentially any number of tasks and not only for $m=n$:

\begin{theorem}\label{th:lower-bound-discrete}
For any number of $m=O(ne^n)$ tasks, there exists an instance of the Bayesian scheduling problem where VCG is not better than $\Omega\left(\frac{\ln n}{\ln \ln n}\right)$-approximate and the tasks have i.i.d. processing times drawn from a discrete MHR distribution.
\end{theorem}
\begin{proof}
Consider an instance with $n$ identical machines and $m$ tasks where the processing times $t_{ij}$ are drawn from  $\sset{0,1}$ such that $t_{ij}=1$ with probability $\left(\frac{n}{2m}\right)^{\frac{1}{n}}\equiv p$ and $t_{ij}=0$ with probability $1-p$. Notice that this is a well-defined distribution, since for all $m\geq n$ we have $p< 1$. Furthermore, it is easy to check that this distribution is MHR; its hazard rate at $0$ is $\frac{\probability{t_{ij}= 0}}{\probability{t_{ij}\geq 0 }}=\frac{1-p}{1}=1-p$ and at $1$ is $\frac{\probability{t_{ij}= 1}}{\probability{t_{ij}\geq 1 }}=\frac{p}{p}=1$.

Next, let $M$ be the random variable denoting the number of tasks whose best
processing time over all machines is non-zero, that is
$$
M=\cards{\ssets{j\fwhs{\min_i t_{ij}=1}}}.
$$
Then $M$ follows a binomial distribution with probability of success $p^n$ and
$m$ trials, since the probability of a task having processing time $1$ at \emph
{all} machines (success) is $p^n$, while there are $m$ tasks in total. Given the
definition for $p$, the average number of tasks that will end up requiring a
processing time of $1$ on every machine is $\expect[M]=mp^n=\frac{n}{2}$.
Also,
we can derive that
$$
\probability{M\geq 3n }\leq e^{-n}
\qquad\text{and}
\qquad
\probability{M\leq \frac{n}{7} }\leq e^{-n/8}
$$
using Chernoff bounds\footnote{Here we use the following forms, with
$\beta_1=1+\sqrt{5}$ and $\beta_2=\frac{\sqrt{2}}{2}$: 
for all $\beta_1 >0$ and $0<\beta_2 <1$,
$$\probability{X\geq(1+\beta_1)\mu}\leq e^{-\frac{\beta_1^2}{2+\beta_1}\mu}
\qquad\text{and}\qquad
\probability{X\leq(1-\beta_2)\mu}\leq e^{-\frac{\beta_2^2}{2}\mu},
$$
for any binomial random variable with mean $\mu$.}.
As we have argued before, we can use classical results from balls-in-bins
analysis to bound the performance of VCG. So, if $\frac{n}{7}<M<3n$, we know
that the expected makespan will be $\varOmega\left(\frac{\ln n}{\ln \ln
n}\right)$, since each task has processing time
$1$ on all machines. That event happens almost surely, with probability at least
$1-e^{-n}-e^
{-n/8}=1-o(1)$.

On the other hand, we next show that the mechanism that simply balances the $M$ ``expensive'' tasks across the machines (by allocating $\left\lceil\frac{M}{n} \right\rceil$ of them to every machine) achieves a constant makespan, hence providing a constant upper-bound on the optimal makespan:
$$
\E[\opt]
\leq 
\probability{M< 3n}\cdot \frac{3n}{n}\cdot 1 + \probability{M\geq 3n}\cdot
\left\lceil \frac{m}{n}\right\rceil \cdot 1
\leq 
3+e^{-n}\left(\frac{m}{n}+1\right)
\leq 4+\frac{m}{ne^{n}}=O(1).$$
\end{proof}
Notice however that Theorem~\ref{th:lower-bound-discrete} still leaves open the possibility for \emph{continuous} MHR distributions to perform better (see also Theorem~\ref{th:VCG_upper_iid_const} and Corollary~\ref{cor:bound4}).

We finally conclude with a couple of simple observations, for the sake of completeness.
First, our initial requirement (see Section~\ref{sec:preliminaries}) for identical machines (which is a standard one, see~\cite{Chawla:2013rm}) is crucial for guaranteeing any non-trivial approximation ratios on the performance of VCG:
\begin{observation}
There exists an instance of the Bayesian scheduling problem where VCG is not better than $n$-approximate even when the tasks are identically distributed according to continuous MHR distributions.
\end{observation}
\begin{proof}
Assume $\frac{m}{n}$ being an integer, and give as input the point-mass
distributions $t_{1j}=1-\varepsilon$ and $t_{ij}=1$ for all $j\in[m]$ and
$i=2,3,\dots,n$, where $\varepsilon\in(0,1)$. Notice that the execution times are indeed identical across the tasks. The $\algoname{VCG}$ mechanism allocates all jobs to machine $1$, for a makespan of $m\cdot (1-\varepsilon)$, while the algorithm that assigns $\frac{m}{n}$ jobs to each machine achieves a makespan of at most $\frac{m}{n}\cdot 1$, resulting to a ratio of $n$ as $\varepsilon \to 0$. Without loss, the above analysis carries over even if we replace the point-mass distributions with uniform distributions over a small interval around the values $1-\varepsilon$  and $1$. These distributions are MHR, which concludes the proof.
\end{proof}

We now present some lower bounds on the performance of the mechanisms analyzed
by \citet{Chawla:2013rm}. A definition of these mechanisms can be found in the introduction.
The following demonstrates that the analysis of the approximation ratio for the class of bounded overload mechanisms presented in \cite{Chawla:2013rm} is asymptotically tight:
\begin{observation}\label{obs:bo-lb-fixed}
For any number of $m\geq n$ tasks, there exists an instance of the Bayesian scheduling problem where a bounded overload mechanism with parameter $c$ is not better than $\min\{c\frac{m}{n},n-1\}$-approximate and the processing times are drawn from machine-identical continuous MHR distributions.
\end{observation}
\begin{proof}
Consider the instance of Theorem~\ref{thm:lb-vcg-identical-machines} and recall that the optimal makespan is equal to $1$. We note that since each task has the same processing time at any machine, all possible allocations such that no machine is assigned to more than $c\frac{m}{n}$ tasks are valid outputs of bounded overload mechanisms with parameter $c$. Now consider the bounded overload mechanism which fixes an ordering of the machines and then breaks ties according to that ordering. This mechanism would allocate at least $\min\{c\frac{m}{n},n-1\}$ unit-cost tasks on the first machine in its ordering.
\end{proof}

The same instance can be used to bound the performance of the bounded overload mechanism with parameter $c$ that breaks ties uniformly at random as well. Having sufficiently many tasks ($m=\Omega\left(\frac{n\ln n}{\ln \ln n}\right)$) implies that the mechanism behaves almost like the VCG mechanism while allocating the unit-cost tasks, assuming they are the first to be allocated. This gives a lower bound of $\Omega\left(\frac{\ln n}{\ln \ln n}\right)$ on the approximation ratio of this mechanism as well.

Similar instances can provide lower bounds on the performance of the class of sieve and bounded overload mechanisms with parameters $c,\beta,$ and $\delta$, even for the case of i.i.d. processing times. To see this notice that if all tasks have $t_{ij}=1$ with probability $1$ on any machine ($T[1:k]=1$ for any $k$), and we choose threshold $\beta<1$ as is done in \cite{Chawla:2013rm} for the case $m\leq n \ln n$, then a sieve and bounded overload mechanism with parameters $c,\beta\leq 1,$ and $\delta$ immediately reduces to a bounded overload mechanism with parameter $c$ on $\delta n $ machines.

\paragraph{Acknowledgements:} We want to thank Elias Koutsoupias for useful discussions.

\bibliographystyle{abbrvnat}
\bibliography{BayesianScheduling}

\begin{thebibliography}{35}
\providecommand{\natexlab}[1]{#1}
\providecommand{\url}[1]{\texttt{#1}}
\expandafter\ifx\csname urlstyle\endcsname\relax
  \providecommand{\doi}[1]{doi: #1}\else
  \providecommand{\doi}{doi: \begingroup \urlstyle{rm}\Url}\fi

\bibitem[Archer and Tardos(2001)]{AT01}
A.~Archer and {\'{E}}.~Tardos.
\newblock Truthful mechanisms for one-parameter agents.
\newblock In \emph{FOCS}, pages 482--491, 2001.

\bibitem[Ashlagi et~al.(2012)Ashlagi, Dobzinski, and Lavi]{ADL12}
I.~Ashlagi, S.~Dobzinski, and R.~Lavi.
\newblock Optimal lower bounds for anonymous scheduling mechanisms.
\newblock \emph{Math. Oper. Res.}, 37\penalty0 (2):\penalty0 244--258, 2012.

\bibitem[Aven(1985)]{Aven:1985kx}
T.~Aven.
\newblock Upper (lower) bounds on the mean of the maximum (minimum) of a number
  of random variables.
\newblock \emph{Journal of Applied Probability}, 22\penalty0 (3):\penalty0 pp.
  723--728, 1985.

\bibitem[Azar et~al.(2014)Azar, Kleinberg, and Weinberg]{Azar:2014}
P.~D. Azar, R.~Kleinberg, and S.~M. Weinberg.
\newblock {Prophet Inequalities with Limited Information}.
\newblock In \emph{Proceedings of the Twenty-Fifth Annual ACM-SIAM Symposium on
  Discrete Algorithms}, SODA '14, pages 1358--1377, 2014.

\bibitem[Barlow et~al.(1963)Barlow, Marshall, and Proschan]{Barlow:1963fk}
R.~E. Barlow, A.~W. Marshall, and F.~Proschan.
\newblock Properties of probability distributions with monotone hazard rate.
\newblock \emph{Ann. Math. Statist.}, 34\penalty0 (2):\penalty0 375--389, 06
  1963.

\bibitem[Berenbrink et~al.(2008)Berenbrink, Friedetzky, Hu, and
  Martin]{Berenbrink:2008ab}
P.~Berenbrink, T.~Friedetzky, Z.~Hu, and R.~Martin.
\newblock On weighted balls-into-bins games.
\newblock \emph{Theoretical Computer Science}, 409\penalty0 (3):\penalty0 511
  -- 520, 2008.

\bibitem[Chawla et~al.(2012)Chawla, Immorlica, and Lucier]{CIL12}
S.~Chawla, N.~Immorlica, and B.~Lucier.
\newblock On the limits of black-box reductions in mechanism design.
\newblock In \emph{STOC}, pages 435--448, 2012.

\bibitem[Chawla et~al.(2013)Chawla, Hartline, Malec, and Sivan]{Chawla:2013rm}
S.~Chawla, J.~D. Hartline, D.~Malec, and B.~Sivan.
\newblock Prior-independent mechanisms for scheduling.
\newblock In \emph{STOC}, pages 51--60, 2013.

\bibitem[Christodoulou and Kov{\'{a}}cs(2013)]{CK13}
G.~Christodoulou and A.~Kov{\'{a}}cs.
\newblock A deterministic truthful {PTAS} for scheduling related machines.
\newblock \emph{{SIAM} J. Comput.}, 42\penalty0 (4):\penalty0 1572--1595, 2013.

\bibitem[Christodoulou et~al.(2009)Christodoulou, Koutsoupias, and
  Vidali]{CKV09}
G.~Christodoulou, E.~Koutsoupias, and A.~Vidali.
\newblock A lower bound for scheduling mechanisms.
\newblock \emph{Algorithmica}, 55\penalty0 (4):\penalty0 729--740, 2009.

\bibitem[Clarke(1971)]{Clarke71}
E.~H. Clarke.
\newblock Multipart pricing of public goods.
\newblock \emph{Public Choice}, 11\penalty0 (1):\penalty0 17--33, 1971.

\bibitem[Daskalakis and Weinberg(2015)]{DW15}
C.~Daskalakis and S.~M. Weinberg.
\newblock Bayesian truthful mechanisms for job scheduling from bi-criterion
  approximation algorithms.
\newblock In \emph{SODA}, pages 1934--1952, 2015.

\bibitem[Devanur et~al.(2011)Devanur, Hartline, Karlin, and Nguyen]{DHKN11}
N.~R. Devanur, J.~D. Hartline, A.~R. Karlin, and C.~T. Nguyen.
\newblock Prior-independent multi-parameter mechanism design.
\newblock In \emph{WINE}, pages 122--133, 2011.

\bibitem[Dhangwatnotai et~al.(2011)Dhangwatnotai, Dobzinski, Dughmi, and
  Roughgarden]{DDDR11}
P.~Dhangwatnotai, S.~Dobzinski, S.~Dughmi, and T.~Roughgarden.
\newblock Truthful approximation schemes for single-parameter agents.
\newblock \emph{{SIAM} J. Comput.}, 40\penalty0 (3):\penalty0 915--933, 2011.

\bibitem[Dhangwatnotai et~al.(2015)Dhangwatnotai, Roughgarden, and Yan]{DRY15}
P.~Dhangwatnotai, T.~Roughgarden, and Q.~Yan.
\newblock Revenue maximization with a single sample.
\newblock \emph{Games and Economic Behavior}, 91:\penalty0 318--333, 2015.

\bibitem[Dughmi et~al.(2012)Dughmi, Roughgarden, and Sundararajan]{DRS12}
S.~Dughmi, T.~Roughgarden, and M.~Sundararajan.
\newblock Revenue submodularity.
\newblock \emph{Theory of Computing}, 8\penalty0 (1):\penalty0 95--119, 2012.

\bibitem[Epstein et~al.(2013)Epstein, Levin, and van Stee]{Epstein2016}
L.~Epstein, A.~Levin, and R.~van Stee.
\newblock {A Unified Approach to Truthful Scheduling on Related Machines}.
\newblock \emph{Mathematics of Operations Research}, 41\penalty0 (1):\penalty0
  1243--1252, 2013.

\bibitem[Giannakopoulos and Kyropoulou(2015)]{gkyr2015-wine}
Y.~Giannakopoulos and M.~Kyropoulou.
\newblock The {VCG} mechanism for bayesian scheduling.
\newblock In E.~Markakis and G.~Schäfer, editors, \emph{Web and Internet
  Economics (WINE)}, volume 9470 of \emph{Lecture Notes in Computer Science},
  pages 343--356. Springer Berlin Heidelberg, 2015.
\newblock \doi{10.1007/978-3-662-48995-6_25}.
\newblock URL \url{http://arxiv.org/abs/1509.07455}.

\bibitem[Goldner and Karlin(2016)]{Goldner2016}
K.~Goldner and A.~R. Karlin.
\newblock {A Prior-Independent Revenue-Maximizing Auction for Multiple Additive
  Bidders}.
\newblock In \emph{WINE}, 2016.

\bibitem[Groves(1973)]{Groves73}
T.~Groves.
\newblock {Incentives in Teams}.
\newblock \emph{Econometrica}, 41\penalty0 (4):\penalty0 617--31, July 1973.

\bibitem[Hall(1997)]{HochbaumHall97}
L.~A. Hall.
\newblock Approximation algorithms for scheduling.
\newblock In D.~S. Hochbaum, editor, \emph{Approximation Algorithms for NP-hard
  Problems}, pages 1--45. PWS, Boston, 1997.

\bibitem[Hartline and Roughgarden(2009)]{HR09}
J.~D. Hartline and T.~Roughgarden.
\newblock Simple versus optimal mechanisms.
\newblock In \emph{EC}, pages 225--234, 2009.

\bibitem[Koutsoupias and Vidali(2013)]{KV13}
E.~Koutsoupias and A.~Vidali.
\newblock A lower bound of 1+\emph{{\(\varphi\)}} for truthful scheduling
  mechanisms.
\newblock \emph{Algorithmica}, 66\penalty0 (1):\penalty0 211--223, 2013.

\bibitem[Lavi and Swamy(2009)]{LS09}
R.~Lavi and C.~Swamy.
\newblock Truthful mechanism design for multidimensional scheduling via cycle
  monotonicity.
\newblock \emph{Games and Economic Behavior}, 67\penalty0 (1):\penalty0
  99--124, 2009.

\bibitem[Lenstra et~al.(1990)Lenstra, Shmoys, and Tardos]{LST90}
J.~K. Lenstra, D.~B. Shmoys, and {\'{E}}.~Tardos.
\newblock Approximation algorithms for scheduling unrelated parallel machines.
\newblock \emph{Math. Program.}, 46:\penalty0 259--271, 1990.

\bibitem[Lu(2009)]{L09}
P.~Lu.
\newblock On 2-player randomized mechanisms for scheduling.
\newblock In \emph{WINE}, pages 30--41, 2009.

\bibitem[Lu and Yu(2008)]{LY08}
P.~Lu and C.~Yu.
\newblock An improved randomized truthful mechanism for scheduling unrelated
  machines.
\newblock In \emph{STACS}, pages 527--538, 2008.

\bibitem[Motwani and Raghavan(1995)]{Motwani:1995aa}
R.~Motwani and P.~Raghavan.
\newblock \emph{Randomized Algorithms}.
\newblock Cambridge University Press, 1995.

\bibitem[Nisan and Ronen(2001)]{Nisan:2001aa}
N.~Nisan and A.~Ronen.
\newblock Algorithmic mechanism design.
\newblock \emph{Games and Economic Behavior}, 35\penalty0 (1/2):\penalty0
  166--196, 2001.

\bibitem[Nisan and Ronen(2007)]{NR07}
N.~Nisan and A.~Ronen.
\newblock Computationally feasible {VCG} mechanisms.
\newblock \emph{J. Artif. Int. Res.}, 29\penalty0 (1):\penalty0 19--47, 2007.

\bibitem[Raab and Steger(1998)]{Raab:1998ab}
M.~Raab and A.~Steger.
\newblock ``{B}alls into bins'' - {A} simple and tight analysis.
\newblock In \emph{RANDOM}, pages 159--170, 1998.

\bibitem[Roughgarden et~al.(2012)Roughgarden, Talgam{-}Cohen, and Yan]{RTY12}
T.~Roughgarden, I.~Talgam{-}Cohen, and Q.~Yan.
\newblock Supply-limiting mechanisms.
\newblock In \emph{EC}, pages 844--861, 2012.

\bibitem[Sivan(2013)]{S13}
B.~Sivan.
\newblock \emph{Prior Robust Optimization}.
\newblock PhD thesis, University of Wisconsin-Madison, 2013.

\bibitem[Vickrey(1961)]{Vickrey61}
W.~Vickrey.
\newblock Counterspeculation, auctions, and competitive sealed tenders.
\newblock \emph{Journal of Finance}, 16\penalty0 (1):\penalty0 8--37, March
  1961.

\bibitem[Yu(2009)]{Y09}
C.~Yu.
\newblock Truthful mechanisms for two-range-values variant of unrelated
  scheduling.
\newblock \emph{Theor. Comput. Sci.}, 410\penalty0 (21-23):\penalty0
  2196--2206, May 2009.

\end{thebibliography}

\appendix
\section{Omitted Proofs from Section \ref{sec:preliminaries}}\label{append:preliminaries}

\begin{customlemma}{\ref{lemma:first_order_statistic_MHR}}
If $T$ is a continuous MHR random variable then for any positive integer $n$, its first order statistic $T[1:n]$ is also MHR.
\end{customlemma}
\begin{proof}
If $T$ is a continuous real random variable with cdf $F$ and pdf $f$, then the cdf and pdf of $T[1:n]$ are $F^{(1)}(x)=1-(1-F(x))^n$ and $f^{(1)}(x)=nf(x)(1-F(x))^{n-1}$, respectively. So, the hazard rate of $T[1:n]$ is
$$
\frac{f^{(1)}(x)}{1-F^{(1)}(x)}=\frac{n(1-F(x))^{n-1}f(x)}{(1-F(x))^n}=n\frac{f(x)}{1-F(x)},
$$
which is increasing since $\frac{f(x)}{1-F(x)}$ is increasing.
\end{proof}

\begin{customlemma}{\ref{lemma:MHR_squares}}
For any \emph{continuous} MHR random variable $X$ and any positive integer $r$, $\expect[X^r]\leq r!\expect[X]^r$.
\end{customlemma}
\begin{proof}
For any positive integer $s$, denote the normalized moments $\lambda_s\equiv\frac{\expect[X^s]}{s!}$. Then from \cite[p.~384]{Barlow:1963fk} we know that for all integers $i$ and $t>s>0$,
$$
\left(\frac{\lambda_{i+t}}{\lambda_i}\right)^s\leq \left(\frac{\lambda_{i+s}}{\lambda_i}\right)^t.
$$
By selecting $t=r$, $s=1$ and $i=0$, this inequality gives $\lambda_r \lambda_0^{r-1}\leq \lambda_1^r$. We get the desired inequality by noticing that $\lambda_r=\expect[X^r]/r!$, $\lambda_1=\expect[X]$ and $\lambda_0=\expect[1]=1$.
\end{proof}
The continuity assumption in Lemma~\ref{lemma:MHR_squares} is essential, as it is demonstrated by the following example: consider a discrete random variable $X$ over $\sset{0,1}$ with $\probability{X=0}=\frac{1}{2}+\varepsilon$ and $\probability{X=1}=\frac{1}{2}-\varepsilon$, for some small $\varepsilon>0$. This distribution is MHR since its hazard rate at $0$ and $1$ respectively is $h(0)=\frac{\probability{X= 0}}{\probability{X\geq 0 }}=\frac{1}{2}+\varepsilon$ and $h(1)=\frac{\probability{X= 1}}{\probability{X\geq 1 }}=1$. However, it is easy to see that $\expect[X^2]=\expect[X]=\probability{X=1}=\frac{1}{2}-\varepsilon$ and thus $\frac{\expect[X]^2}{\expect[X^2]}=\expect[X]<\frac{1}{2}$.

\section{Proof of Corollary \ref{cor:exp_uniform}}
\label{append:expo_uniform_computations}
Throughout this section we will use the fact that if $T$ is a random variable with cdf $F$ then for any positive integer $n$, the cdf's of the first and last order statistics $T[1:n]$ and $T[n:n]$ are  given by
\begin{equation}
\label{eq:order_statistics_pdf}
F^{(1)}(x)=1-(1-F(x))^n \quad\text{and}\quad F^{(n)}(x)=F^{n}(x),
\end{equation}
respectively.
\begin{lemma}
\label{lemma:stretch-uniform}
If $T$ is a uniform random variable over $[0,1]$, then for all positive integers $n,m$
$$
\expect[T[1:n]]=\frac{1}{n+1}
\quad\text{and}\quad
\expect[T[1:n][m:m]]= 1-mB\left(m,1+\frac{1}{n}\right),
$$
where $B(x,y)\equiv \int_0^1 t^{x-1}(1-t)^{y-1}\,dt$ is the beta function.
\end{lemma}
\begin{proof}
If $T$ is a uniformly distributed random variable over $[0,1]$ then its cdf is given by $F(x)=x$, $x\in[0,1]$.
The first equality is very easy, since from~\eqref{eq:order_statistics_pdf} the cdf of $T[1:n]$ is $1-(1-x)^n$, thus its expectation is $\int_0^1(1-x)^n\,dx=\frac{1}{n+1}$. For the second one, again from~\eqref{eq:order_statistics_pdf} it is straightforward to see that the cdf of $T[1:n][m:m]$ is $[1-(1-x)^n]^m$ so its expectation is $\int_0^1 1-[1-(1-x)^n]^m\,dx=1-\int_0^1 [1-(1-x)^n]^m\,dx$. Next we compute the value of this integral
$$I(m)\equiv \int_0^1 [1-(1-x)^n]^m\,dx.$$
We have:
\begin{align*}
I(m) 	&=\int_0^1[1-(1-x)^n]^{m-1}(1-(1-x)^n)\,dx\\
		&=I(m-1)-\int_0^1[1-(1-x)^n]^{m-1}(1-x)^{n}\,dx\\
		&=I(m-1)-\frac{1}{n}\int_0^1[1-(1-x)^n]^{m-1}\left(1-(1-x)^{n}\right)'(1-x)\,dx\\
		&=I(m-1)-\frac{1}{nm}\int_0^1\left([1-(1-x)^n]^{m}\right)'(1-x)\,dx\\
		&=I(m-1)-\frac{1}{nm}\left[\left(1-(1-x)^n\right)^{m}(1-x)\right]_{x=0}^{x=1}+\frac{1}{nm}\int_0^1[1-(1-x)^n]^{m}(1-x)'\,dx\\
		&=I(m-1)-\frac{1}{nm}\int_0^1[1-(1-x)^n]^{m}\,dx\\
		&=I(m-1)-\frac{1}{nm}I(m),
\end{align*}
meaning that
$$
I(m)=\frac{1}{1+\frac{1}{nm}}I(m-1)\quad\text{with}\quad I(1)=\int_0^1 1-(1-x)^n\,dx=1-\frac{1}{n+1}.
$$
Solving the above recurrence gives
$$
I(m)=\frac{2\cdot 3\cdot \dots \cdot m}{\left(1+\frac{1}{n}\right)\cdot \left(2+\frac{1}{n}\right)\cdot \dots \cdot \left(m+\frac{1}{n}\right)}=\frac{m\varGamma(m)\varGamma\left(1+\frac{1}{n}\right)}{\varGamma\left(m+1+\frac{1}{n}\right)}=mB\left(m,1+\frac{1}{n}\right),
$$
where $\varGamma$ denotes the (complete) gamma function.
\end{proof}

\begin{lemma}
\label{lemma:stretch-expo}
If $T$ is an exponentially distributed random variable with parameter $\lambda$, then for all positive integers $n,m$
$$
\expect[T[1:n]]=\frac{1}{\lambda n}
\quad\text{and}\quad
\expect[T[1:n][m:m]]= \frac{H_m}{\lambda n},
$$
where $H_m=1+\frac{1}{2}+\dots+\frac{1}{m}$ is the harmonic function.
\end{lemma}
\begin{proof}
If $T$ is exponentially distributed, then its cdf is given by $F(x)=1-e^{-\lambda x}$, $x\in[0,\infty)$, where $\lambda $ is a positive real parameter.
The first equality is again easy, since from~\eqref{eq:order_statistics_pdf} the cdf of $T[1:n]$ is $1-(e^{-\lambda x})^n$, thus its expectation is $\int_0^\infty e^{-n\lambda\cdot x}\,dx=\frac{1}{\lambda n}$. For the second one,
from~\eqref{eq:order_statistics_pdf} it is straightforward to see that the cdf
of $T[1:n][m:m]$ is $(1-e^{-\lambda n x})^m$ so its expectation is
\begin{align*}
\int_0^\infty 1-(1-e^{-\lambda n x})^m\,dx 	&= \frac{1}{\lambda n}\int_0^\infty 1-(1-e^{-y})^m\,dy, &&\text{by changing $y=\lambda n x$},\\
		&= \frac{1}{\lambda n}\int_0^1 \frac{1-(1-z)^m}{z}\,dz, &&\text{by changing $z=e^{-y}$},\\
		&= \frac{1}{\lambda n}\int_0^1 \frac{1-w^m}{1-w}\,dw, &&\text{by changing $w=1-z$},\\
		&= \frac{1}{\lambda n}\int_0^1 \sum_{k=0}^{m-1}w^k\,dw\\
		&= \frac{1}{\lambda n}\sum_{k=0}^{m-1}\frac{1}{k+1}\\
		&=	\frac{H_m}{\lambda n}.
\end{align*}
\end{proof}

To conclude the proof of Corollary~\ref{cor:exp_uniform}, from Lemma~\ref{lemma:stretch-uniform} we deduce that the stretch factor of the uniform distribution is
$$
\frac{\expect[T[1:n][n:n]]}{\expect[T[1:n]]}=(n+1)\left[1-nB\left(n,1+\frac{1}{n}\right)\right]
$$
and from Lemma~\ref{lemma:stretch-expo} the stretch factor for the exponential distribution with parameter $\lambda$ is
$$
\frac{\expect[T[1:n][n:n]]}{\expect[T[1:n]]}=H_{n}.
$$
It can be verified that both the above quantities are lower-bounded by $\ln n$.
\end{document}